\documentclass{llncs}

\usepackage{stmaryrd}

\usepackage{latexsym}
\usepackage{amstext,amssymb,amsmath}
\usepackage{prftree}
\usepackage{thmtools,thm-restate}

\usepackage{mdframed}

\title{On Stronger Calculi for QBFs\thanks{The work was supported by
    the Austrian Science Foundation (FWF) under grant
    S11409-N23. Partial results have been announced at the QBF
    Workshop 2014
    (http://www.easychair.org/smart-program/VSL2014/QBF-program.html).
}
}

\author{Uwe Egly}

\institute{
Institut f\"ur Informationssysteme 184/3,
Technische Universit\"at Wien, \\
Favoritenstrasse 9--11, A-1040 Vienna, Austria\\
 {email: \texttt{uwe@kr.tuwien.ac.at}}
}

\newcommand{\base}{{\sf Base}}
\newcommand{\ih}{{\sf IH}}
\newcommand{\step}{{\sf Step}}

\newcommand{\sccase}[1]{\noindent {\sc Case~#1}}
\newcommand{\scsubcase}[2]{\noindent {\sc Subcase~#1.#2}}

\newcommand{\emptyclause}{\Box}

\newcommand{\qres}{\mbox{\sf Q-res}}
\newcommand{\qures}{\mbox{\sf QU-res}}
\newcommand{\qdres}{\mbox{\sf Q(D)-res}}
\newcommand{\ldqres}{\mbox{\sf LDQ-res}}
\newcommand{\ldqures}{\mbox{\sf LDQU-res}}
\newcommand{\ldqupres}{\mbox{\sf LDQU{}$^+$-res}}

\newcommand{\ForallRed}{\mbox{\sf $\forall$R}}

\newcommand{\Rfo}{$\mbox{\sf R}_1$}

\newcommand{\IRcalc}{$\mbox{\sf IR-calc}$}
\newcommand{\IRcalcPM}[2]{$\mbox{\sf IR-calc}(#1,#2)$}
\newcommand{\IRcalcPMSubst}[3]{$\mbox{\sf IR-calc}(#1,#2,#3)$}
\newcommand{\IRcalcPMDSubst}[4]{$\mbox{\sf IR-calc}(#1,#2,#3,#4)$}
\newcommand{\IRMcalc}{$\mbox{\sf IRM-calc}$}

\newcommand{\Res}{$\mbox{\sf Res}$}
\newcommand{\Fac}{$\mbox{\sf Fac}$}
\newcommand{\Axiom}{$\mbox{\sf Axiom}$}
\newcommand{\Inst}{$\mbox{\sf Inst}$}
\newcommand{\InstOp}[2]{$\mbox{\sf inst}(#1,#2)$}

\newcommand{\Gqvetree}{$\mbox{{\sf Gqve}}^*$}

\newcommand{\quantifier}{{\sf Q}}
\newcommand{\propvars}{{\cal PV}}

\newcommand{\falsum}{\bot}
\newcommand{\verum}{\top}

\newcommand{\dom}[1]{\ensuremath\mathit{dom}(#1)}
\newcommand{\range}[1]{\ensuremath\mathit{rg}(#1)}

\newcommand{\sequence}[1]{\overline{#1}}

\newcommand{\complassign}{\ensuremath{\veebar}}
\newcommand{\is}{\ensuremath{\backslash}}

\newcommand{\restrict}[1]{\ensuremath{[#1]}}

\newcommand{\caC}{{\cal C}}

\newcommand{\caQ}{{\cal Q}}
\newcommand{\caT}{{\cal T}}
\newcommand{\caU}{{\cal U}}

\newcommand{\impl}{\rightarrow}
\newcommand{\lequiv}{\leftrightarrow}
\newcommand{\xor}{\oplus}

\newcommand{\lc}[1]{\mathit{lc}(#1)}

\newcommand{\size}[1]{\mathit{size}(#1)}

\newcommand{\dep}[1]{\mathit{dep}(#1)}

\newcommand{\var}[1]{\mathit{var}(#1)}

\newcommand{\hkb}{\ensuremath{\Psi}}
\newcommand{\hkbt}[1]{\ensuremath{\hkb_{#1}}}

\newcommand{\QBFtoPL}[3]{\llbracket\,#1\,\rrbracket_{#2}^{#3}}

\newcommand{\SkQBFtoPL}[3]{\mathit{Sk}\llbracket\,#1\,\rrbracket_{#2}^{#3}}

\newcommand{\EPRQBFtoPL}[3]{\mathit{EPR}\llbracket\,#1\,\rrbracket_{#2}^{#3}}

\newcommand{\level}[1]{{\mathit lv(#1)}}

\newcommand{\PHPText}{CPHP}
\newcommand{\PHPTextp}{DPHP}

\newcommand{\TPHPTextp}{TPHP}

\newcommand{\PHPVar}[2]{\mbox{$\mathsf{\PHPText}_{#1}^{#2}$}}
\newcommand{\PHPpVar}[2]{\mbox{$\mathsf{\PHPTextp}_{#1}^{#2}$}}

\newcommand{\TPHPpVar}[2]{\mbox{$\mathsf{\TPHPTextp}_{#1}^{#2}$}}

\renewcommand{\qed}{\hfill\ensuremath{\square}}

\newcommand{\card}[1]{\left\vert{#1}\right\vert}

\allowdisplaybreaks

\newcommand{\THROWIT}[1]{}

\begin{document}

\maketitle
 
\begin{abstract}
Quantified Boolean formulas (QBFs) generalize propositional formulas
by admitting quantifications over propositional variables.  QBFs can
be viewed as (restricted) formulas of first-order predicate logic and
easy translations of QBFs into first-order formulas exist.  We analyze
different translations and show that first-order resolution combined
with such translations can polynomially simulate well-known deduction
concepts for QBFs. Furthermore, we extend QBF calculi by the
possibility to instantiate a universal variable by an existential
variable of smaller level. Combining such an enhanced calculus with
the propositional extension rule results in a calculus with a
universal quantifier rule which essentially introduces propositional
formulas for universal variables. In this way, one can mimic a very
general quantifier rule known from sequent 
systems.
\end{abstract}

\section{Introduction}\label{sec:intro}

Quantified Boolean formulas (QBFs) generalize propositional formulas
by admitting quantifications over propositional variables.  QBFs can
be viewed in two different ways, namely (i) as a generalization of
propositional logic and (ii) as a restriction of first-order predicate
logic (where we interpret over a two element domain). A number of
calculi are available for QBFs: the ones based on variants of
resolution for QBFs
\cite{DBLP:journals/iandc/BuningKF95,DBLP:conf/cp/Gelder12,DBLP:journals/fmsd/BalabanovJ12,BWJ:SAT14},
the ones based on instantiating universal variables with truth
constants combined with  propositional resolution and an additional
instantiation rule \cite{DBLP:conf/mfcs/BeyersdorffCJ14}, and
different sequent systems
\cite{DBLP:journals/aml/CookM05,Krajicek:1995,ESW-cj-2009,DBLP:conf/sat/Egly12}.

In all these calculi (except the latter ones from
\cite{DBLP:journals/aml/CookM05,Krajicek:1995,DBLP:conf/sat/Egly12}),
the possibility to instantiate a given formula is limited. In purely
resolution-based calculi, formulas (or more precisely universal
variables) are never instantiated. In instantiation-based calculi,
instantiation is restricted to truth constants. In contrast, sequent
systems possess flexible quantifier rules, and (existential) variables
as well as (propositional) formulas can be used for instantiation with
tremendous speed-ups in proof complexity. This motivates why we are
interested in strengthening instantiation techniques for
instantiation-based calculi.

We allow to replace (some) universal variables $x$ not only by truth
constants but by existential variables left of $x$ in the quantifier
prefix. This approach mimics the effect of quantifier rules
introducing atoms in sequent calculi from \cite{DBLP:conf/sat/Egly12}.
We add a propositional extension principle (known from extended
resolution \cite{Tseitin:1968}), which enables the introduction of
propositional formulas for universal variables via extension variables
(or names for the formula).  Contrary to \cite{DBLP:conf/sat/Egly12},
where we proposed propositional extensions of the form $\exists q
(q\lequiv F)$ which can be eliminated if the cut rule is available in
the sequent calculus, such an elimination is not possible here for
which reason we have to use (classical) extensions.

\smallskip
\noindent
\emph{Contributions.}
\vspace*{-0.5em}
\begin{enumerate}
\item We consider different translations from QBFs to first-order
  logic \cite{DBLP:conf/cade/SeidlLB12} and provide a
  proof-theoretical analysis of the translation in combination with
  first-order resolution (\Rfo). We exponentially separate two
  variants of the translation in Theorem~\ref{thm:EPR-notpsim-SK}.

\item We show that such combinations can polynomially simulate
  Q-resolution with resolution over existential and universal
  variables (\qures{} \cite{DBLP:conf/cp/Gelder12}, 
  Theorem~\ref{thm:R1-psim-QUres}), Q-resolution (\qres{}
  \cite{DBLP:journals/iandc/BuningKF95}, 
  Corollary~\ref{cor:Rfo-psim-qres}) and the instantiation-based
  calculus \IRcalc{} \cite{DBLP:conf/mfcs/BeyersdorffCJ14}
  (Theorem~\ref{Rfo-psim-IRcalc},
  Corollary~\ref{cor:Rfo-EPR-psim-IRcalc}). The latter simulation
  provides a soundness proof for \IRcalc{} independent from strategy
  extraction.

\item We show in Theorem~\ref{thm:usual-calculi-npsim-Rfo} that
  neither \qres{} nor \qures, the long-distance Q-resolution variants
  \ldqres, \ldqures, \ldqupres{}
  \cite{DBLP:conf/iccad/ZhangM02,DBLP:journals/fmsd/BalabanovJ12,BWJ:SAT14},
  different in\-stan\-tiation-based calculi
  \cite{DBLP:conf/mfcs/BeyersdorffCJ14} nor \qdres{}
  \cite{SlivovskySzeider-SAT14} can polynomially simulate \Rfo{} with
  one of the considered translations.

\item We generalize \IRcalc{} by the possibility to instantiate
  universal variables not only with truth constants but also with
  existential variables (similar to the corresponding quantifier rule
  in \cite{DBLP:conf/sat/Egly12}). We show in
  Proposition~\ref{prop:IRcalcPM-npsim-IRcalcPMSubst} that this
  generalized calculus is actually stronger than the original one.

\item We combine generalized \IRcalc{} by a propositional extension
  rule \cite{Tseitin:1968,BCJ-AAAI-WS16}
  essentially enabling the introduction of Boolean functions (instead
  of atoms and truth constants) for universal variables.

\end{enumerate}

\noindent
\emph{Structure.}  In Sect.~\ref{sec:prel} we introduce necessary
definitions and notations. In Sect.~\ref{sec:diff-trans} different
translations from QBFs to (restrictions of) first-order logic
\cite{DBLP:conf/cade/SeidlLB12} are reconsidered. In
Sect.~\ref{sec:res-calculi} different calculi based on (variants of)
the resolution calculus are described. Here, we introduce our calculi
generalized from \IRcalc. In Sect.~\ref{sec:poly-sim} we present our
results on polynomial simulations between considered calculi and in
Sect.~\ref{sec:exp-sep} we provide exponential separations. In the
last section we conclude and discuss future research possibilities.

\section{Preliminaries}\label{sec:prel}

We assume familiarity with the syntax and semantics of propositional
logic, QBFs and first-order logic (see, e.g.,
\cite{Leitsch-resolution1997} for an introduction). We recapitulate
some notions and notations which are important for the rest of the
paper.

We consider a propositional language based on a set $\propvars$ of
Boolean variables and truth constants $\verum$ (true) and $\falsum$
(false), both of which are not in $\propvars$. A variable or a truth
constant is called \emph{atomic} and connectives are from $\{\neg,
\land, \lor, \impl, \lequiv, \xor\}$.
A \emph{literal} is a variable or its negation.  A \emph{clause} is a
disjunction of literals, but sometimes we consider it as a set of
literals. \emph{Tautological clauses} contain a variable and its
negation and the \emph{empty clause} is denoted by $\Box$.
Propositional formulas are denoted by capital Latin letters like $A,
B, C$ possibly annotated with subscripts, superscripts or primes.

We extend the propositional language by Boolean quantifiers.
Universal ($\forall^b$) and existential ($\exists^b$) quantification
is allowed within a QBF.  The superscript $b$ is used to distinguish
Boolean quantifiers from first-order quantifiers introduced later.
QBFs are denoted by Greek letters.  Observe that we allow non-prenex
formulas, i.e., quantifiers may occur deeply in a QBF.
An example for a non-prenex QBF is
$\forall^b p\, (p \impl \forall^b q \exists^b r \, (q\land r\land
s))$, where $p$, $q$, $r$ and $s$ are variables.  Moreover, free
variables (like $s$) are allowed, i.e., there might be occurrences of
variables in the formula for which we have no quantification. Formulas
without free variables are called \emph{closed}; otherwise they are
called \emph{open}. The \emph{universal} (\emph{existential}) closure
of $\varphi$ is $\forall^b x_1 \ldots \forall^b x_n \varphi$
($\exists^b x_1 \ldots \exists^b x_n \varphi$), for which we often
write $\forall^b \vec{X} \varphi$ ($\exists^b \vec{X} \varphi$) if
$\vec{X} = \{x_1, \ldots ,x_n\}$ is the set of all free variables in
$\varphi$.  A formula in \emph{prenex conjunctive normal form} (PCNF)
has the form $Q_1^bp_1\ldots Q_n^bp_n \, M$, where $Q_1^bp_1\ldots
Q_n^bp_n$ is the \emph{quantifier prefix}, $Q\in \{ \forall, \exists
\}$ and $M$ is the (propositional) \emph{matrix} which is in CNF.
Often we write a QBF as $Q_1^bX_1\ldots Q_k^bX_k \, M$ ($Q_i \neq
Q_{i+1}$ for all $i=1, \ldots , k-1$ and the elements of $\{X_1,
\ldots , X_k\}$ are pairwise disjoint). We define the \emph{level of a
  literal $\ell$}, $\level{\ell}$, as the index $i$ such that the
variable of $\ell$ occurs in $X_i$.  The \emph{logical complexity} of a
formula $\Phi$, $\lc{\Phi}$, is the number of occurrences of
connectives and quantifiers.

We use a first-order language consisting of (objects) {\em variables},
\emph{function symbols} (FSs), \emph{predicate symbols} (PSs),
together with the truth constants and connectives mentioned above.
Quantifiers $\forall$ and $\exists$ bind object variables.  \emph{
  Terms} and \emph{formulas} are defined according to the usual
formation rules.  We identify $0$-ary PSs with
propositional atoms, and $0$-ary FSs with
\emph{constants}. Clauses, tautological clauses and the empty clause
are defined as in the propositional case.

Let $V$ be the set of first-order variables and $T$ be the set of
terms.  A \emph{substitution} is a mapping $\sigma$ of type
$V\rightarrow T$ such that $\sigma(v) \neq v$ only for finitely many
variables $v\in V$.  We represent $\sigma$ by a finite set of the form
$\{v_1\backslash t_1, \ldots , v_n \backslash t_n\}$.  The
\emph{domain} of 
$\sigma$, $\dom{\sigma}$, is the set
$\{v \mid v\in V, \sigma(v) \neq v\}$. The \emph{range} of $\sigma$,
$\range{\sigma}$, is the set $\{\sigma(v) \mid v\in
\dom{\sigma}\}$. We call $\sigma$ a \emph{variable substitution} if
$\range{\sigma} \subseteq V$.  The \emph{empty} substitution
$\epsilon$ is denoted by $\{\}$. We often write substitutions
post-fix, e.g., we use $x\sigma$ instead of $\sigma(x)$.
Algebraically, substitutions define a monoid with $\epsilon$ being the
neutral element under the usual composition of substitutions.

Substitutions are extended to terms and formulas in the usual way,
e.g., $f(t_1,\ldots , t_n)\sigma = f(t_1\sigma,\ldots , t_n\sigma)$,
$(\neg) p(t_1,\ldots , t_n)\sigma = (\neg) p(t_1\sigma,\ldots ,
t_n\sigma)$, and $(F\circ G)\sigma = F\sigma \circ G\sigma$, where $f$
is an $n$-place FS, $p$ is an $n$-place PS, $t_1, \ldots , t_n$ are
terms, $F$ and $G$ are (quantifier-free) formulas and $\circ$ is a
binary connective.  For substitutions $\sigma$ and $\tau$, $\sigma$ is
\emph{more general than} $\tau$ if there is a substitution $\mu$ such
that $\sigma\mu = \tau$.
A substitution $\sigma$ is called a \emph{permutation} if $\sigma$ is
one-one and a variable substitution. A permutation $\sigma$ is called
a \emph{renaming} (substitution) of an expression $E$ (i.e., $E$ is a
term or a quantifier-free formula) if $\var{E}\cap \range{\sigma} =
\{\}$, where $\var{E}$ is the set of all variables occurring in $E$.
For an expression $G$, $G\sigma$ is a \emph{variant} of $G$ provided
$\sigma$ is a renaming substitution.

Let $E=\{E_1, \ldots , E_n\}$ be a non-empty set of expressions. A
substitution $\sigma$ is called a \emph{unifier of $E$} if
$\card{\{E_1\sigma, \ldots , E_n\sigma\}} = 1$. Unifier $\sigma$ is
called \emph{most general unifier} (mgu), if for every unifier $\tau$
of $E$, $\sigma$ is more general than $\tau$.

Let $P_1$ and $P_2$ be two proof systems.  $P_1$
\emph{polynomially simulates} (p-simulates) $P_2$ if there is a
polynomial $p$ such that, for every natural number $n$ and every
formula $\Phi$, the following holds. If there is a proof of
$\Phi$ in $P_2$ of size $n$, then there is a proof of $\Phi$ (or
a suitable translation of it) in $P_1$ whose size is less than $p(n)$.

\section{Different translations of QBFs to first-order logic}
\label{sec:diff-trans}

We introduce different translations of (closed) QBFs to (closed)
formulas in (restrictions of) first-order logic.  We start with the
basic translation from \cite{DBLP:conf/cade/SeidlLB12} in
Fig.~\ref{fig:translation1}.  Obviously, the QBF $\Phi$ and the
first-order formula $\QBFtoPL{\Phi}{p}{f}$ enjoy a very similar
structure. Especially the variable dependencies expressed by the
quantifier prefix are exactly the same.

\begin{figure}[t]
\begin{tabular}{ccccc}
$\QBFtoPL{\falsum}{p}{f} = p(f_{0})$ & \hspace{1em} &
$\QBFtoPL{\verum}{p}{f} = p(f_{1})$  & \hspace{1em} &
$\QBFtoPL{x}{p}{f} = p(x)$       
\\[1.5em]
\multicolumn{1}{c}{
$\QBFtoPL{\neg \Phi}{p}{f} = 
\neg \QBFtoPL{\Phi}{p}{f}$
} & &
\multicolumn{1}{c}{
$\QBFtoPL{\Phi_1 \circ \Phi_2}{p}{f} = 
\QBFtoPL{\Phi_1}{p}{f} \circ \QBFtoPL{\Phi_2}{p}{f}$} 
& & 
\multicolumn{1}{c}{
$\QBFtoPL{\quantifier^b x\, \Phi}{p}{f} = 
\quantifier x\, \QBFtoPL{\Phi}{p}{f}$
}\\
\end{tabular}
\caption{\label{fig:translation1} The translation of QBFs to
  first-order formulas. The connective $\circ$ is a binary connective
  present in both languages and $\quantifier \in \{\forall,
  \exists\}$.  The symbols $p$ and $f$ do not occur in the source QBF;
  $p$ is a unary predicate symbol and $f$ is used to construct
  constant and function symbols by indices.  }
\end{figure}

\begin{restatable}{prop}{QBFIsomorphic}
\label{prop:isomorphic}%
Let $\Phi$ be a (closed) QBF and let $\QBFtoPL{\Phi}{p}{f}$ be its
(closed) first-order translation. Then $\Phi \cong
\QBFtoPL{\Phi}{p}{f}$, i.e., $\Phi$ and $\QBFtoPL{\Phi}{p}{f}$ are
isomorphic. 
\end{restatable}
\noindent
The proof in the appendix is by induction on the logical complexity of
$\Phi$.

The basic translations from Fig.~\ref{fig:translation1} can be
extended to $\SkQBFtoPL{\Phi}{p}{f}$ generating a skolemized version
of $\QBFtoPL{\Phi}{p}{f}$. We restrict our attention here to QBFs in
PCNF.

\begin{definition}
Let $\Phi$ be a closed QBF in PCNF with matrix $M$ and let
$\QBFtoPL{\Phi}{p}{f}$ be its closed first-order translation.  For any
existential variable $a$ in the quantifier prefix of $\Phi$, let
$\dep{a}$ be the sequence of universal variables left of $a$ (in
exactly the same order in which they occur in the prefix). Let $f_a$
be the Skolem function symbol associated to $a$.  We call
$\QBFtoPL{M}{p}{f} \sigma$ the \emph{skolemized form} of
$\QBFtoPL{M}{p}{f}$ and denote it by $\SkQBFtoPL{M}{p}{f}$, where the
substitution $\sigma$ is as follows.
$$
\sigma = \{ a \backslash f_a(\dep{a}) \mid \text{for all existential variables
  $a$ in $\Phi$}\}
$$
\end{definition}
Traditionally, $\SkQBFtoPL{M}{p}{f}$ is denoted as a quantifier-free
formula with the assumption that all free variables are (implicitly)
universally quantified. 

The number of universal variables a Skolem function depends on can be
optimized, e.g., by using miniscoping or dependency schemes
\cite{DBLP:conf/cade/SeidlLB12}. As we will see later on, most of our
results do not depend on such optimizations.

\begin{proposition}\label{prop:matrix-isomorphic}
Let $\Phi$ be a closed QBF in PCNF with matrix $M$ and let
$\QBFtoPL{\Phi}{p}{f}$ be its closed first-order translation.  Let
$\SkQBFtoPL{M}{p}{f}$ be the \emph{skolemized form} of
$\QBFtoPL{M}{p}{f}$. Then $M \cong \SkQBFtoPL{M}{p}{f}$. 
\end{proposition}

Due to propositions~\ref{prop:isomorphic} and
\ref{prop:matrix-isomorphic},
we can relate each literal of each clause from $M$ to \emph{its}
isomorphic counterpart in $\QBFtoPL{M}{p}{f} \sigma$.

Since we interpret over a two-element domain, proper Skolem function
symbols (i.e., the arity is greater than $0$) can be eliminated by
introducing new predicate symbols. The resulting formula belongs to
EPR (Effectively PRopositional logic or more traditionally it belongs
to the Bernays-Schoenfinkel class).

\begin{definition}
Let $\Phi$ be a closed QBF in PCNF with matrix $M$ and let
$\QBFtoPL{\Phi}{p}{f}$ be its closed first-order translation.  Let
$\SkQBFtoPL{M}{p}{f}$ the skolemized form of $\QBFtoPL{M}{p}{f}$.
Replace any occurrence of a predicate of the form $p(f_b(X))$ by
$p_b(X))$ where $f_b$ is a proper function symbol 
and $X$ is a non-empty list of universal
variables. The formula resulting after all possible replacements is
the EPR formula $\EPRQBFtoPL{M}{p}{f}$.
\end{definition}
We will see later that the first-order 
and the EPR translation have different proof-theoretical properties
because some
resolutions are blocked by different predicate symbols.
Proposition~\ref{prop:sat-equiv} is Lemma~1 in
\cite{DBLP:conf/cade/SeidlLB12} (stated without a proof).

\begin{restatable}{prop}{SatEquiv}
\label{prop:sat-equiv}
  Let $\Phi$ be a closed QBF.  Then
\begin{eqnarray*}
  \text{$\Phi$ is satisfiable} & \text{\quad iff\/ \quad} & 
  \text{$\QBFtoPL{\Phi}{p}{f}\land
    p(f_{1}) \land \neg p(f_{0})$ is satisfiable.} 
\end{eqnarray*}
\end{restatable}
\noindent
A proof can be found in the appendix.

\begin{figure}[tb]
\begin{mdframed}
\begin{center}
\begin{tabular}{ccccccc}
\prfbyaxiom{\Axiom}{C} &
\hspace{2em} &
\prftree[r]{\Res}
{x \lor C_1}{\neg x\lor C_2}{C_1\lor C_2} &
\hspace{2em} &
\prftree[r]{\Fac}
{C\lor \ell \lor \ell}{C\lor \ell} &
\hspace{2em} &
\prftree[r]{\ForallRed}
{D\lor m}{D}
\end{tabular}
\end{center}

\medskip

$C$ is a non-tautological clause from the matrix. If $y\in C_1$
then $\neg y\notin C_2$. Variable $x$ is existential (\qres) and
existential or universal (\qures), $\ell$ is a literal and $m$ is
a universal literal. If $e\in D$ is existential, then $\level{e} <
\level{m}$ holds.
\end{mdframed}
\caption{The rules of \qres{} and \qures{} 
\cite{DBLP:journals/iandc/BuningKF95,DBLP:conf/cp/Gelder12}
\label{fig:QUres}}
\end{figure}

\begin{figure}[tbh]
\begin{mdframed}
\begin{displaymath}
\prfbyaxiom{\Axiom}{\{e^{\restrict{\sigma}} \mid e \in C,
\text{$e$ is existential}\}}
\end{displaymath}

\medskip

$C$ is a non-tautological clause from the matrix $M$, $\sigma = \{u\is
0 \mid \text{$u\in C$ universal}\}$ where $u\is 0$ is a shorthand for
$x\is 0$ if $u=x$ and $x\is 1$ if $u=\neg x$.

\begin{center}
\begin{tabular}{ccccc}
\prftree[r]{\Res}
{x^{\tau} \lor C_1}{\neg x^{\tau}\lor C_2}{C_1\lor C_2} &
\hspace{4em} &
\prftree[r]{\Fac}
{C\lor \ell^{\tau} \lor \ell^{\tau}}{C\lor \ell^{\tau}} &
\hspace{4em} &
\prftree[r]{Inst}{C}{\text{\InstOp{\tau}{C}}}
\end{tabular}
\end{center}
$\tau$ is an assignment to universal variables and
$\range{\tau}\subseteq \{0,1\}$.
\end{mdframed}
\caption{The rules of  \IRcalc(P,M){} taken from
\cite{DBLP:conf/mfcs/BeyersdorffCJ14}\label{fig:IR-CalcPM}}
\end{figure}

\section{Different calculi based on resolution}\label{sec:res-calculi}
We introduce different calculi used in this paper.  We start with two
resolution calculi, \qres{} and \qures, for QBFs in
Fig.~\ref{fig:QUres}.  Observe that the consequence of each 
rule is non-tautological.  We continue with the calculus
\IRcalcPM{P}{M}{} in Fig.~\ref{fig:IR-CalcPM}, where we use the same
presentation as in \cite{DBLP:conf/mfcs/BeyersdorffCJ14}.  $P$ is the
quantifier prefix and $M$ is the quantifier-free matrix in CNF. In 
the following instantiation-based calculi, inference rules do not work
on usual clauses but on \emph{annotated clauses} based on
\emph{extended assignments}. An extended assignment is a
partial mapping from the Boolean variables to $\{0,1\}$. An annotated
clause consists of \emph{annotated literals} of the form
$\ell^{\restrict{\tau}}$, where $\tau$ is an extended assignment to
\emph{universal} variables and $\restrict{\tau} = \{u\is c \mid (u\is
c)\in \tau, \level{u}< \level{\ell}\}$ with $c\in \{0,
1\}$. Composition of extended assignments is defined using
\emph{completion}. The expression $\mu \complassign \tau$ is called
the completion of $\mu$ by $\tau$.  Then $\sigma$, the completion of
$\mu$ by $\tau$, is defined as follows.
\begin{equation}\label{eqn:completion}
\sigma(x) = 
\begin{cases}
\mu(x)  & \text{if $x\in \dom{\mu}$};\\
\tau(x) & \text{if $x\notin \dom{\mu}$ and $x\in \dom{\tau}$}.
\end{cases}
\end{equation}
The function \InstOp{\tau}{C} allows instantiations of clauses; it
computes $\{\ell^{\restrict{\mu \complassign \tau}} \mid \ell^{\mu} \in C \}$ 
for an extended assignment $\tau$ and an annotated clause $C$.
Later on, we will clarify the
relation between annotations and substitutions in first-order logic.

We extend \IRcalcPM{\cdot}{\cdot}{} by the possibility to instantiate
universal variables by existential ones.  Technically the
instantiation is performed by a global substitution $\sigma_v$.  If a
universal variable $x$ is replaced by some existential variable $e$,
i.e., $(x\is e) \in \sigma_v$, then $\level{e}<\level{x}$ must
hold. We name the calculus equipped with the substitution $\sigma_v$
\IRcalcPMSubst{P}{M}{\sigma_v} and depict the rules in
Fig.~\ref{fig:IR-CalcPM-enhanced}.

\begin{figure}[tb]
\begin{mdframed}
\begin{displaymath}
\prfbyaxiom{\Axiom}{\{e^{\restrict{\sigma}} \mid e \in C\sigma_v,
\text{$e$ is existential}\}}
\end{displaymath}

\begin{enumerate}
\item $C$ is a non-tautological clause from the matrix $M$.
\item $\sigma_v = \{x \is e \mid \text{$x$ is universal, 
 $e$ is existential, $\level{e}<\level{x}$}\}$.
\item $C\sigma_v = \{e \mid \text{$e\in C$ existential}\} \cup 
\{x\sigma_v \mid \text{$x\in C$ universal}, x\in \dom{\sigma_v}\} 
\cup \mbox{}$ \\
\hspace*{4em}$\{x \mid \text{$x\in C$ universal}, x\notin \dom{\sigma_v}\}$.
\item $\sigma$, \Res, \Fac{} and \Inst{} are the same as 
in \IRcalcPM{\cdot}{\cdot}.
\end{enumerate}
\end{mdframed}
\caption{The rules of  \IRcalcPMSubst{P}{M}{\sigma_v}
\label{fig:IR-CalcPM-enhanced}}
\end{figure}
It is immediately apparent that 
this calculus is sound and complete. We get completeness, when we use
the empty substitution as $\sigma_v$ because then,
\IRcalcPMSubst{\cdot}{\cdot}{\cdot} reduces to \IRcalcPM{\cdot}{\cdot}
which is sound and complete \cite{DBLP:conf/mfcs/BeyersdorffCJ14}.
Soundness follows from the validity of QBFs of the form
\begin{displaymath}
\caQ_1\exists e \caQ_2 \forall x \caQ_3 \, \varphi(e,x) \, \impl \, 
\caQ_1\exists e \caQ_2 \caQ_3 \, \varphi(e,e).
\end{displaymath}
If the right formula has an \IRcalcPM{\cdot}{\cdot} refutation, then
it is false and therefore the left formula has to be false.

\begin{figure}[t]
\begin{mdframed}
\begin{displaymath}
\prfbyaxiom{\Axiom}{\{e^{\restrict{\sigma}} \mid e \in C\sigma_v,
\text{$e$ is existential}\}}
\end{displaymath}

\begin{enumerate}
\item $C$ is a non-tautological clause from the matrix $M$ or from $\Delta$.

\item $\sigma_v$, $C\sigma_v$, $\sigma$, \Res, \Fac{} and \Inst{} are
  the same as in \IRcalcPMSubst{\cdot}{\cdot}{\cdot}.

\item If $C\in \Delta$ then $\sigma = \epsilon$ and $C=C\sigma_v$ by
  construction.

\end{enumerate}
\end{mdframed}
\caption{The rules of  \IRcalcPMDSubst{P}{M}{\Delta}{\sigma_v}
\label{fig:IR-CalcPMD-enhanced}}
\end{figure}

We further enhance \IRcalcPMSubst{\cdot}{\cdot}{\cdot}{} by the
possibility to use propositional extensions
\cite{Tseitin:1968,BCJ-AAAI-WS16}.  This extension operation is a
generalization of the well-known structure-preserving translation to
(conjunctive) normal form in propositional logic.  For presentational
reasons, we require to have all extensions at the very beginning of
the deduction in order to allow extension variables as replacements
for universal variables. Figure~\ref{fig:IR-CalcPMD-enhanced} shows
the inference rules of this calculus
\IRcalcPMDSubst{P}{M}{\Delta}{\sigma_v}, where $\Delta$ is a sequence
$\delta_1,\ldots , \delta_d$ of (clausal representations of)
extensions of the form $\delta_i\colon q_i \lequiv F$ with $F$ being
of the form $\neg p$ or of the form $p\circ r$ ($\circ\in \{\land,
\lor, \impl, \lequiv, \xor\}$) and $q_i$ is a variable neither
occurring in $M$ nor in $F$ nor in $\delta_1, \ldots ,
\delta_{i-1}$. The variables $q_i, p, r$ are existential.  The
quantification $\exists q_i$ extends the quantifier prefix $P$ such
that $\level{v} \leq \level{q_i}$ for all variables $v$ occurring in
$F$ and $\level{q_i}$ is minimal. Due to the requirements on the
extension variables $q_i$ and the placement of $\exists q_i$, the
resulting calculus is sound. Completeness is not an issue here,
because we can use an empty $\Delta$.

\begin{remark}\label{rem:clauses-as-sets}
The usual formalization of clauses and resolvents as sets of literals
can be simulated in our formalizations by the factoring rule
\Fac. We assume in the following that \Fac{} is applied
as soon as possible.
\end{remark}

We finally introduce first-order resolution.  Let $C$ be a clause and
let $K$ and $L$ be two distinct literals in $C$ both of which are
either negated or unnegated.  If there is an mgu $\sigma$ of $K$ and
$L$, then the clause $D=C\sigma=\{N\sigma \mid N\in C\}$ is called a
\emph{factor} of $C$. The clause $C$ is called the \emph{premise} of the
factoring operation.

Let $C$ and $D$ be two clauses and let $D'$ be a variant of $D$ which
has no variable in common with $C$. A clause $E$ is a \emph{resolvent}
of the parent clauses $C$ and $D$ if the following conditions hold:
\begin{enumerate}
\item $K\in C$ and $L'\in D'$ are literals of opposite sign whose
  atoms are unifiable by an mgu $\sigma$.

\item $E=\big(C\sigma \setminus \{K\sigma\}\big) \cup 
\big(D'\sigma \setminus \{L'\sigma\}\big)$.
\end{enumerate}
Let $\caC$ be a set of clauses. A sequence $C_1, \ldots , C_n$ is
called \emph{\Rfo{} deduction} (first-order resolution deduction) of a
clause $C$ from $\caC$ if $C_n = C$ and for all $i = 1, \ldots , n$,
one of the following conditions hold.
\begin{enumerate}
  \item $C_i$ is an input clause from $\caC$.
  \item $C_i$ is a factor of a $C_j$ for $j < i$.
  \item $C_i$ is a resolvent of $C_j$ and $C_k$ for $j,k < i$.
\end{enumerate}
An \emph{\Rfo{} refutation of $\caC$} is an \Rfo{} deduction of the
empty clause $\emptyclause$ from $\caC$.  The \emph{size} of a
deduction is given by $\sum_{i=1}^n \size{C_i}$, where $\size{C_i}$ is
the number of character occurrences in $C_i$. An \Rfo{} deduction has
\emph{tree form} if every occurrence of a clause is used at most once
as a premise in a factoring operation or as a parent clause in a
resolution operation.

Next we introduce the \emph{subsumption rule} taken from
Definition~2.3.4 in \cite{Eder:1992}. Contrary to the usual use of
subsumption in automated deduction as a deletion rule, here we
\emph{add} clauses which are (factors of) instantiations of clauses.

\begin{definition}
If $C$ and $D$ are clauses, then \emph{$C$ subsumes $D$} or \emph{$D$
  is subsumed by $C$}, if there is a substitution $\sigma$ such that
$C\sigma \subseteq D$. A set $S'$ of clauses is obtained from a set
$S$ by subsumption if $S'=S\cup \{D\}$ where $D$ is subsumed by a
clause of $S$.
\end{definition}

Resolution can be extended by the subsumption rule
(Definition~3.2.3 in \cite{Eder:1992}). 
\begin{definition}
By a derivation of a set of clauses $S_2$ from a set of clauses $S_1$
by \emph{\Rfo{} plus subsumption}, we mean a sequence $C_1, \ldots ,
C_n$ of clause such that the following conditions are fulfilled.
\begin{enumerate}
\item $S_2 \subseteq S_1 \cup \{C_1, \ldots ,C_n\}$.
\item \label{item:RfoPlusSubs}
For all $k=1,\ldots , n$ there is a clause $C\in S_1 \cup \{C_1,
  \ldots ,C_{k-1}\}$ subsuming the clause $C_k$ or there exist clauses
  $C, D \in S_1 \cup \{C_1, \ldots ,C_{k-1}\}$ such that $C_k$ is
  subsumed by a resolvent of $C$ and $D$.
\end{enumerate}
\end{definition}
Factors are not needed in item~\ref{item:RfoPlusSubs}, because the
factor of $C$ can be generated by subsumption.  We need a simplified
version of Proposition~3.2.1 from \cite{Eder:1992}.

\begin{proposition}\label{prop:R1subs-psim-R1}
\Rfo{} polynomially simulates  \Rfo{} plus subsumption.
\end{proposition}
The subsumption rule is not necessary but makes proofs of polynomial
simulation results much more convenient.  It allows instantiated
deductions for which eventually the lifting theorem provides a
deduction ``on the most general level''.

\section{Polynomial simulations of calculi} 
\label{sec:poly-sim} 

In this section we show that \Rfo{} together with a suitable
translation $\caT$ (denoted by \Rfo{} $+$ $\caT$) polynomially
simulates \qures, \qres{} and \IRcalcPM{\cdot}{\cdot}.

\begin{restatable}{thm}{RfoPsimQU}
\label{thm:R1-psim-QUres}
\Rfo{} $+$ $\SkQBFtoPL{\cdot}{p}{f}$ polynomially
simulates \qures.
\end{restatable}
\noindent 
The proof is by induction on the number of clauses in the \qures{}
deduction. It can be found in the appendix. It shows that first-order
literals obtained from universal literals in the QBF and eliminated by
\ForallRed{} are eliminated by resolutions with $p(f_1)$ and $\neg
p(f_0)$ \emph{without} instantiating the first-order resolvent.

\begin{corollary}\label{cor:Rfo-psim-qres}
The following results are immediate consequences of 
Theorem~\ref{thm:R1-psim-QUres}.
\begin{enumerate}
\item \Rfo{} $+$ $\EPRQBFtoPL{\cdot}{p}{f}$ 
polynomially simulates \qures.

\item \Rfo{} $+$ $\SkQBFtoPL{\cdot}{p}{f}$ as well as \Rfo{} $+$
  $\EPRQBFtoPL{\cdot}{p}{f}$ polynomially simulates \qres.

\end{enumerate}
\end{corollary}

We present a soundness proof of \IRcalcPM{\cdot}{\cdot}{} independent
from strategy extraction by a polynomial simulation of
\IRcalcPM{\cdot}{\cdot}{} by \Rfo.

\begin{definition}
\label{def:composition-subst}
Let $\tau = \{x_1\is s_1, \ldots , x_k \is s_k\}$ and $\mu =
\{y_1\is t_1, \ldots , y_l\is t_l\}$ be two substitutions. The
\emph{composition} of $\tau$ and $\mu$, $\tau\mu$, is obtained
from
\begin{displaymath}
\big\{x_1\is s_1\mu, \ldots , x_k \is s_k\mu, 
y_1\is t_1, \ldots , y_l\is t_l \big\} 
\end{displaymath}
by deleting all $y_i\is t_i$ for which $y_i \in \{x_1, \ldots , x_k\}$
holds.
\end{definition}

\begin{lemma}\label{lem:compl-is-composition}
Let $\tau$ and $\mu$ be two substitutions as defined in
Definition~\ref{def:composition-subst}, where $x_1,\ldots , x_k, y_1,
\ldots , y_l$ are universal variables and $\{s_1,\ldots , s_k, t_1,
\ldots , t_l\} \subseteq \{0, 1\}$. Then $\tau \complassign \mu$
is the composition $\tau\mu$.
\end{lemma}

\begin{proof}
Let $\sigma$ be the completion of $\tau$ by $\mu$ defined in
(\ref{eqn:completion}).
Since $\dom{\tau}$ as well as $\dom{\mu}$ is a subset of the set of
universal variables and $\range{\tau}$ as well as $\range{\mu}$ is a
subset of $\{0,1\}$, $\range{\tau} \cap \dom{\mu} = \{\}$ and
therefore $s_i\mu = s_i$ for all $i=1,\ldots , k$. Hence, the
completion $\sigma$ of the two substitutions $\tau$ and $\mu$ is
exactly their composition $\tau\mu$.\qed
\end{proof}

In the following, we deal with annotated clauses $C$ of the form
$\{l_1^{\restrict{\sigma_1}}, \ldots , l_k^{\restrict{\sigma_k}}\}$
where any $l_i$ is an existential literal and any
$\restrict{\sigma_i}$ is the restriction of assignment $\sigma_i$ to
exactly those universal variables $x\in \dom{\sigma_i}$ for which
$\level{x}<\level{l_i}$ holds. We denote the sequence of \emph{all}
universal variables $x$ with $\level{x}<\level{l_i}$ by $\dep{l_i} =
\sequence{X}_{l_i}$ where we assume the same order as in the
quantifier prefix. A first-order clause $D$ corresponding to $C$ is
constructed as follows
\begin{displaymath}
\big\{(\neg) p(f_e(\sequence{X}_e)) \sigma \mid (\neg)
e^{\restrict{\sigma}} \in C \quad \text{and} \quad 
p(f_e(\sequence{X}_e))\cong e \big\}, 
\end{displaymath}
where $p(f_e(\sequence{X}_e))$ is the isomorphic counterpart of $e$
(cf.\ the remark after Proposition~\ref{prop:matrix-isomorphic}).
Using $\sequence{X}_e$ together with $\sigma$ mimics the effect of
$\restrict{\sigma}$; the difference is the explicit notation of
\emph{all} universal variables $\sequence{X}_{e}$ left of $e$ and not
only the variables in $\sequence{X}_{e}\cap \dom{\sigma}$.

\begin{restatable}{thm}{RfoPsimIRcalc}
\label{Rfo-psim-IRcalc}
\Rfo{} $+$ $\SkQBFtoPL{\cdot}{p}{f}$
polynomially simulates \IRcalcPM{\cdot}{\cdot}.
\end{restatable}

In the proof, we construct by induction on the number of derived
clauses in the \IRcalc{} deduction stepwisely a deduction in \Rfo{}
plus subsumption. We consider the sequence of first-order clauses
obtained from the original clauses as a skeleton for the final
proof. Since the clauses in the skeleton do not follow by a single
application of an inference rule, we have to provide a short deduction
of the clauses.

\begin{proof}
We utilize Proposition~\ref{prop:R1subs-psim-R1} and allow subsumption
in the simulation.  The proof is by strong mathematical induction on
the number of derived clauses in the \IRcalc{} deduction. Let $P(n)$
denote the statement ``Given a \IRcalc{} deduction $C_1, \ldots , C_n$
from a QBF $Q.M$ and a sequence of first-order clauses $D_1, \ldots ,
D_n$, the clause $D_n$ has a short deduction in \Rfo{} plus
subsumption from $p(f_1), \neg p(f_0), \SkQBFtoPL{M}{p}{f}, D_1,
\ldots , D_{n-1}$''.

\medskip\noindent \base: $n=1$. $C_1$ is a consequence of the axiom
rule using clause $C$ from the matrix $M$. Let $\sigma$ be the
assignment induced by $C$. Then we have a clause $D\in
\SkQBFtoPL{M}{p}{f}$ from which we can derive $D_1\sigma$ by resolution
steps using $p(f_1)$ and $\neg p(f_0)$. The number of these steps is
equal to the number of universal variables in $C$.

\medskip\noindent \ih: Suppose $P(1), \ldots , P(n)$ hold for some
$n\geq 1$.

\medskip\noindent \step: We have to show $P(n+1)$. Consider $C_1,
\ldots , C_{n+1}$ and $D_1, \ldots , D_{n+1}$.

\medskip\noindent \sccase{1}: $C_{n+1}$ is derived by the axiom
rule. Then proceed like in the base case.

\medskip\noindent \sccase{2}: $C_{n+1}$ is a consequence of the
rule \Inst{} with premise $C_i$ (for some $i$ with $1\leq i\leq n$)
and assignment $\tau$. By \ih{} and Remark~\ref{rem:clauses-as-sets},
we have a short \Rfo{} plus subsumption deduction of $D_i = \{(\neg)
p(f_e(\sequence{X}_e)) \sigma \mid (\neg) e^{\restrict{\sigma}} \in
C_i\}$. $C_{n+1}$ is of the form $\{(\neg)
e^{\restrict{\sigma\complassign \tau}} \mid (\neg)
e^{\restrict{\sigma}} \in C_i\}$.  By
Lemma~\ref{lem:compl-is-composition}, $x(\sigma\complassign\tau) =
x\sigma\tau$ for any universal variable $x$ with
$\level{x}<\level{e}$. Therefore $D_{n+1}$ is of the form $\{(\neg)
p(f_e(\sequence{X}_e)) \sigma\tau \mid (\neg)
e^{\restrict{\sigma\complassign \tau}} \in C_{n+1}\}$. Now $D_{n+1} =
D_i\tau$ and $D_{n+1}$ can be derived by subsumption.

\medskip\noindent \sccase{3}: $C_{n+1}$ is a consequence of the
rule \Fac{} with premise $C_i: \widetilde{C}_i \lor \ell^{\tau} \lor
\ell^{\tau}$ (for some $i$ with $1\leq i\leq n$).  By \ih, we have a
short \Rfo{} plus subsumption deduction of $D_i\colon \widetilde{C}_i
\lor L \lor L$, where $L$ is of the form $(\neg)
p(f_e(\sequence{X}_e)\tau$. We generate a factor $D_{n+1}$ of $D_i$
  simply by omitting one of the duplicates. 

\medskip\noindent \sccase{4}: $C_{n+1}$ is a consequence of the
resolution rule with parent clauses $C_i, C_j$ (for some $i,j$ with
$1\leq i,j\leq n$). By \ih, we have two clauses 
\begin{displaymath}
D_i = \{p(f_e(\sequence{X}_e)) \sigma\} \cup D_i' 
\qquad  \text{and} \qquad  
D_j =  \{\neg p(f_e(\sequence{X}_e)) \sigma\} \cup D_j'
\end{displaymath}
We use $\lambda$ of the form $\{x\is y\}$ as a renaming of the
variables in $D_j$ such that $D_j\lambda$ does not share any variable
with $D_i$.  The resolvent is $D_i' \cup D_j'\lambda\mu$ where $\mu$ is
the mgu of the form $\{y\is x\mid x\notin \dom{\sigma}\}$. We add
  $D_i' \cup D_j'\lambda\mu\lambda'$ by subsumption, where $\lambda'$
  maps all remaining variables $y$ to their $x$ counterpart.  \qed
\end{proof}

\begin{corollary}\label{cor:Rfo-EPR-psim-IRcalc}
\Rfo{} $+$ $\EPRQBFtoPL{\cdot}{p}{f}$
polynomially simulates \IRcalcPM{\cdot}{\cdot}.
\end{corollary}

When we inspect the translation of (axiom) clauses, we observe
that a universal variable $x$ is translated to an atom of the form
$p(x)$.  With the subsumption rule we can instantiate the clause by a
substitution of the form $\{x\is t\}$ for a term $t$.  This
observation was the trigger to introduce the stronger calculus
\IRcalcPMSubst{\cdot}{\cdot}{\cdot}, where universal variables cannot
be replaced only by $0$ or $1$ but also by any existential variable
$e$ with $\level{e} < \level{x}$.

\section{Exponential separation of resolution calculi} 
\label{sec:exp-sep} 

We constructed in \cite{DBLP:conf/sat/Egly12} a family
$(\Phi_n)_{n\geq 1}$ of short closed QBFs in PCNF for which any
\qres{} refutation of $\Phi_n$ is superpolynomial. We recapitulate the construction here. The formula $\Phi_n$
is
\begin{equation}\label{eqn:Phi_n}
\exists^b X_n \forall^b Y_n \exists^b Z_n
\big( \TPHPpVar{n}{Y_n,Z_n}  \land \PHPVar{n}{X_n} \big)\enspace .
\end{equation}
$\PHPVar{n}{X_n}$ is the pigeon hole formula for $n$ holes and $n+1$
pigeons in \emph{conjunctive} normal form and denoted over the
variables $X_n = \{x_{1,1} , \ldots , x_{n+1, n}\}$. Variable
$x_{i,j}$ is intended to denote that pigeon $i$ is sitting in hole
$j$. \PHPVar{n}{X_n} is
\begin{equation}
  \nonumber
  \bigg(\bigwedge_{i=1}^{n+1} \big(\bigvee_{j=1}^{n} x_{i,j} \big) \bigg) \land 
  \bigg(\bigwedge_{j=1}^{n} \bigwedge_{1\leq i_1<i_2\leq n+1} 
  (\neg x_{i_1,j} \lor \neg x_{i_2,j})\bigg)\enspace .
\end{equation}
The number of clauses in \PHPVar{n}{X_n} is $l_n = (n+1) + n^2(n+1)/2$
and $\size{\PHPVar{n}{X_n}}$ is $O(n^3)$. The formula
\TPHPpVar{n}{Y_n,Z_n} is obtained from the pigeon hole formula in
\emph{disjunctive} normal form, \PHPpVar{n}{Y_n}, by a
structure-preserving polarity-sensitive translation to clause form
\cite{DBLP:journals/jsc/PlaistedG86}. The formula \PHPpVar{n}{Y_n} is
simply the negation of $\PHPVar{n}{Y_n}$ where negation has been
pushed in front of atoms and double-negation elimination has been
applied.

We use new variables of the form $z_{i_1,i_2,j}$ for disjuncts in
\PHPpVar{n}{Y_n}.  For the first $n+1$ disjuncts of the form
$\bigwedge_{j=1}^{n} \neg y_{i,j}$
with $1\leq i\leq n+1$, we use variables $z_{1,0,0}, \ldots ,
z_{n+1,0,0}$. For the second part, for any $1\leq j\leq n$ and the
$n(n+1)/2$ disjuncts, we use
\begin{equation}
\label{eqn:varset2}
z_{1,2,j}, \, \ldots \, , z_{1,n+1,j}, \, z_{2,3,j}, \, \ldots \, , 
z_{2,n+1,j}, \, \ldots \, , z_{n,n+1,j}\enspace .
\end{equation}
The set of these variables for $\PHPpVar{n}{\mbox{}}$ is denoted by
$Z_n$.  Due to this construction, we can speak about the conjunction
corresponding to the variable $z_{i_1,i_2,j}$.

We construct the conjunctive normal form $\TPHPpVar{n}{Y_n,Z_n}$ of
$\PHPpVar{n}{Y_n,Z_n}$ as follows.  First, we take the clause
$D_n^{Z_n} = \bigvee_{z\in Z_n} \, \neg z$ over all variables in
$Z_n$.
The formula $P_n^{Y_n,Z_n}$ for the first $(n+1)$
disjuncts of \PHPpVar{n}{Y_n} is of the form
\begin{equation}
\nonumber
  \bigwedge_{i=1}^{n+1} \bigwedge_{j=1}^{n} (z_{i,0,0} \lor \neg y_{i,j})
\enspace . 
\end{equation}
For the remaining $n^2(n+1)/2$ disjuncts of \PHPpVar{n}{Y_n}, we have
the formula $Q_n^{Y_n,Z_n}$
\begin{eqnarray*}
& &  
\bigwedge_{j=1}^{n} \bigwedge_{1\leq i_1<i_2\leq n+1} 
\big((z_{i_1,i_2,j} \lor y_{i_1,j}) \land
  (z_{i_1,i_2,j} \lor y_{i_2,j})\big)\enspace .
\end{eqnarray*}
Then \TPHPpVar{n}{Y_n,Z_n} is $D_n^{Z_n} \land P_n^{Y_n,Z_n} \land
Q_n^{Y_n,Z_n}$ and $\size{\TPHPpVar{n}{Y_n,Z_n}}$ is $O(n^3)$.  
It is easy to check that $\PHPpVar{n}{Y_n} \lequiv \exists^b Z_n \,
\TPHPpVar{n}{Y_n,Z_n} $ is valid.

Let us  modify the quantifier prefix of $\Phi_n$. By
quantifier shifting rules we get, in an
``antiprenexing'' step, the equivalent formula
$(\forall^b Y_n \exists^b Z_n \TPHPpVar{n}{Y_n,Z_n}) \land 
( \exists^b X_n \PHPVar{n}{X_n})$.
Prenexing yields the equivalent QBF $\Omega_n$
\begin{equation}\label{eqn:Omega_n}
\forall^b Y_n \exists^b Z_n \exists^b X_n \big( \TPHPpVar{n}{Y_n,Z_n}
\land \PHPVar{n}{X_n} \big)
\end{equation}
which has only one quantifier alternation instead of two.  In
\cite{DBLP:conf/sat/Egly12} we showed that $\Phi_n$ and $\Omega_n$
have short cut-free tree proofs in a sequent system \Gqvetree, where
weak quantifiers introduce atoms.  The following extends Proposition~3
in \cite{DBLP:conf/sat/Egly12}.

\begin{proposition}\label{prop:SAT12:prop3}
Any \qres{} refutation of $\Phi_n$ from (\ref{eqn:Phi_n}) and
$\Omega_n$ from (\ref{eqn:Omega_n}) has superpolynomial size.
\end{proposition}
The proof is based on the fact that (i) the two conjuncts belong to
languages with different alphabets and (ii) that the alphabets cannot
be made identical by instantiation of quantifiers in \qres. Therefore
we have to refute either \TPHPpVar{n}{Y_n,Z_n} or \PHPVar{n}{X_n}
under the given quantifier prefix.  Since $\forall^b Y_n \exists^b
Z_n\TPHPpVar{n}{Y_n,Z_n}$ is true, there is no \qres{} refutation and
we have to turn to $\exists^b X_n \PHPVar{n}{X_n}$.  But then, we
essentially have to refute \PHPVar{n}{X_n} with propositional
resolution and consequently, by Haken's famous result
\cite{DBLP:journals/tcs/Haken85}, any \qres{} refutation of
$\PHPVar{n}{X_n}$ is superpolynomial in $n$.

Since \qures, \ldqres, \ldqures, \ldqupres, and Q(D)-resolution
(\qdres) \cite{SlivovskySzeider-SAT14} are based on the same
quantifier-handling mechanism as \qres, the following corollary is
obvious.

\begin{corollary}
Any refutation of $\Phi_n$ from (\ref{eqn:Phi_n}) and $\Omega_n$ from
(\ref{eqn:Omega_n}) in the \qures, \ldqres, \ldqures, \ldqupres, or
\qdres{} calculus has superpolynomial size. 
\end{corollary}
For \IRcalcPM{\cdot}{\cdot} the situation is not better. Since
universal literals are only replaced by $0$, no unification of the two
alphabets can happen.

\begin{proposition}
Any refutation of $\Phi_n$ from (\ref{eqn:Phi_n}) and $\Omega_n$ from
(\ref{eqn:Omega_n}) in \IRcalcPM{\cdot}{\cdot} has size superpolynomial in $n$.
\end{proposition}
The quantifier prefix is unfortunate if one expects $\Omega_n$ being
false. Actually, the initial universal quantifier block prevents any
non-empty $\sigma_v$ and consequently, any
\IRcalcPMSubst{\cdot}{\cdot}{\cdot} refutation of $\Omega_n$ reduces
to an \IRcalcPM{\cdot}{\cdot} refutation of $\Omega_n$.

\begin{proposition}
Any refutation of $\Omega_n$ from (\ref{eqn:Omega_n}) in
\IRcalcPMSubst{\cdot}{\cdot}{\cdot} has size superpolynomial in $n$.
\end{proposition}
In the following we show that 
$\SkQBFtoPL{\Omega_n}{p}{f}$ has a short refutation in \Rfo. We use
$f_{x_{i,j}}$ to denote the Skolem function symbol corresponding to
$x_{i,j} \in X_n$ and $f_{z_{i,j,k}}$ to denote the Skolem function
symbols corresponding to $z_{i,j,k}\in Z_n$. All the Skolem function
symbols have arity $|Y_n| = n(n+1)$.  Let $\overline{F}$ denote the
formula $F$ under the first-order translation.  We have
\begin{align*}
\overline{\PHPVar{n}{X_n}} &\colon 
  \bigg(\bigwedge_{i=1}^{n+1} \big(\bigvee_{j=1}^{n} 
  p(f_{x_{i,j}}(Y_n) \big)  \bigg)  \land \\
&  \qquad\quad
  \bigg(\bigwedge_{j=1}^{n} \bigwedge_{1\leq i_1<i_2\leq n+1} 
  (\neg p(f_{x_{i_1,j}}(Y_n)) \lor \neg p(f_{x_{i_2,j}}(Y_n)))\bigg)\enspace .
\\
\overline{D_n^{Z_n}} 
& \colon 
\bigvee_{z\in Z_n} \, \neg p(f_z(Y_n))
\\
\qquad\qquad 
\overline{P_n^{Y_n,Z_n}}
& \colon 
\bigwedge_{i=1}^{n+1} \bigwedge_{j=1}^{n} (p(f_{z_{i,0,0}}(Y_n)) \lor 
\neg p(y_{i,j}))
\\
\overline{Q_n^{Y_n,Z_n}}
& \colon 
\bigwedge_{j=1}^{n} \bigwedge_{1\leq i_1<i_2\leq n+1} 
\big((p(f_{z_{i_1,i_2,j}}(Y_n)) \lor p(y_{i_1,j})) \land \\
& \hspace*{10em}
  (p(f_{z_{i_1,i_2,j}}(Y_n)) \lor p(y_{i_2,j}))\big)\enspace .
\end{align*}

The refutation of $\SkQBFtoPL{\Omega_n}{p}{f}$ is constructed as
follows.
\begin{enumerate}
\item We use $\overline{P_n^{Y_n,Z_n}}$ together with the first $n+1$
  clauses from $\overline{\PHPVar{n}{X_n}}$ to derive
  $p(f_{z_{i,0,0}}(Y_n))\mu_i$ (for all $i=1,\ldots , n+1$). The
  deduction consists of $O(n^2)$ clauses and applies resolution and
  factoring.  The substitution $\mu_i$ is $\bigcup_{j=1}^n
  \{y_{i,j}\backslash f_{x_{i,j}}(Y_n) \sigma_{i,j}\}$, where
  $\sigma_{i,j}$ is a variable renaming from the variant generation in
  resolution.

\item We use $\overline{Q_n^{Y_n,Z_n}}$ together with the binary
  clauses from $\overline{\PHPVar{n}{X_n}}$ to derive
  $p(f_{z_{i_1,i_2,j}}(Y_n))\nu_{i_1, i_2,j}$ (for all $j=1,\ldots ,
  n$ and $i_1,i_2$ with $1\leq i_1<i_2\leq n+1$).  The deduction
  consists of $O(n^3)$ clauses and applies resolution and factoring.
Then $\nu_{i_1, i_2,j}$ is
  $\{y_{i_1,j}\backslash f_{x_{i_1,j}}(Y_n)\sigma_{i_1, i_2,j},
  y_{i_2,j}\backslash f_{x_{i_2,j}}(Y_n)\sigma_{i_1, i_2,j}\}$. Again
  $\sigma_{i_1, i_2,j}$ is a variable renaming like above.

\item We use $\overline{D_n^{Z_n}}$ together with the derived instance
  of $p(f_{z_{k,l,m}}(Y_n))$ to derive $\emptyclause$ by resolution.
  Since any variable $y_{i,j}$ is assigned to a variant of
  $f_{x_{i,j}}(Y_n)$ for all $i=1,\ldots , n+1$ and all $j=1, \ldots ,
  n$, all resolution steps are possible. The deduction consists of
  $O(n^3)$ clauses.
\end{enumerate}
The formula $\SkQBFtoPL{\Phi_n}{p}{f}$ can be refuted in a similar
fashion in \Rfo{} by replacing variants of the form $f_{x_{i,j}}(Y_n)$
by Skolem constants $a_{i,j}$.

\begin{proposition}\label{prop:Psi_n-short}
Let $(\Phi_n)_{n\geq 1}$ and $(\Omega_n)_{n\geq 1}$ be the families of
closed QBFs defined above.
Then $\QBFtoPL{\Phi_n}{p}{f}$ and $\QBFtoPL{\Omega_n}{p}{f}$ have short
\emph{tree} refutations in \Rfo{} consisting of $O(n^3)$
clauses. Moreover the size of the refutation is $O(n^8)$.
\end{proposition}

\begin{restatable}{thm}{NoPsimRfo}
\label{thm:usual-calculi-npsim-Rfo}
The calculi \qures, \ldqres, \ldqures, \ldqupres{}, \qdres,
\IRcalcPM{\cdot}{\cdot}, \IRcalcPMSubst{\cdot}{\cdot}{\cdot}{} and
\IRMcalc{} cannot polynomially simulate tree \Rfo{} $+$ 
$\SkQBFtoPL{\cdot}{p}{f}$ or \Rfo{} $+$ $\EPRQBFtoPL{\cdot}{p}{f}$.
\end{restatable}

We use $(\Omega_n)_{n\geq 1}$ to exponentially separate \Rfo{}
combined with the two translations, i.e., we compare
$\SkQBFtoPL{\cdot}{p}{f}$ with $\EPRQBFtoPL{\cdot}{p}{f}$.

\begin{proposition}
Let $(\Omega_n)_{n\geq 1}$ be the family of closed QBFs defined above
and let $\Omega'_n$ be the EPR formula
$\EPRQBFtoPL{\Omega_n}{p}{f}\land p(f_{1}) \land \neg p(f_{0})$.  Then
$\Omega'_n$ has only refutation in \Rfo{} of size superpolynomial in
$n$.
\end{proposition}
\begin{proof}[Sketch]
Similar arguments as in Proposition~\ref{prop:SAT12:prop3} apply,
because the EPR translations of $\TPHPpVar{n}{Y_n,Z_n}$ and
$\PHPVar{n}{X_n}$ are denoted in different languages and literals from
the former cannot be resolved with literals from the latter. Again,
the pigeonhole formula has to be refuted. Consequently, the
(essentially propositional) resolution proof has size superpolynomial
in $n$.  \qed
\end{proof}

\begin{restatable}{thm}{ERPNotPsimSK}
\label{thm:EPR-notpsim-SK}
\Rfo{} $+$ $\EPRQBFtoPL{\cdot}{p}{f}$ cannot polynomially
simulate \emph{tree} \Rfo{} $+$ $\SkQBFtoPL{\cdot}{p}{f}$.
\end{restatable}

Let us reconsider the family $(\hkb)_{t \geq 1}$ of QBFs from
\cite{DBLP:journals/iandc/BuningKF95}.  Formula $\hkbt{t}$ 
has the prefix $P_t\colon
\exists d_0 d_1 e_1 \forall x_1 \exists d_2 e_2 \forall x_2 \exists
d_3 e_3\ldots \forall x_{t-1} \exists d_t e_t \forall x_t \exists
f_{1}\ldots f_{t}
$
and the matrix $M_t$ consisting of the following clauses:
\begin{displaymath}
\begin{array}{lclclclcl}
C_0     & \colon & \overline{d}_0 & \hspace{0.5em} & 
C_1     & \colon & d_0 \lor \overline{d}_1 \lor \overline{e}_1 \\
C_{2j}   & \colon & d_j \lor \overline{x}_j \lor \overline{d}_{j+1} 
\lor \overline{e}_{j+1} & & 
C_{2j+1} & \colon & e_j \lor x_j \lor \overline{d}_{j+1} 
\lor \overline{e}_{j+1}
& \hspace{0.5em} & j=1,\ldots , t-1\\
C_{2t}   & \colon & d_t \lor \overline{x}_{t} \lor 
\overline{f}_{1}\lor \cdots  \lor \overline{f}_{t}
& &  
C_{2t+1} & \colon & e_t \lor x_{t} \lor \overline{f}_{1}\lor 
\cdots  \lor \overline{f}_{t}\\
B_{2j}  & \colon & \overline{x}_{j+1} \lor f_{j+1} & & 
B_{2j+1} & \colon & x_{j+1} \lor f_{j+1} & & j=0, \ldots , t-1
\end{array}
\end{displaymath}
By Theorem 3.2 in \cite{DBLP:journals/iandc/BuningKF95} and Theorem~6
in \cite{DBLP:conf/stacs/BeyersdorffCJ15}, any \qres{} refutation and
any \IRcalcPM{\cdot}{\cdot} refutation of \hkbt{t} is exponential
in $t$.  The formula $\hkbt{t}$ has a polynomial size Q-resolution
refutation if universal pivot variables are
allowed~\cite{DBLP:conf/cp/Gelder12}.

Let us extract Herbrand functions from such a short \qures{}
refutation of \hkbt{t} with the method of
\cite{DBLP:journals/fmsd/BalabanovJ12} resulting in $\overline{d_i}
\land e_i$ for $x_i$. We explain in the following how we can produce
short \IRcalcPMDSubst{P_t}{M_t}{\Delta}{\sigma_v}{} refutations using
such 
functions.

Let $\Delta \colon \delta_1,\ldots , \delta_t$ where $\delta_i$ is
$q_i \lor d_i \lor \overline{e}_i, \overline{q}_i \lor \overline{d}_i,
\overline{q}_i \lor e_i$, i.e., $\delta_i$ is the clausal
representation of $q_i \lequiv \overline{d}_i\land e_i$. The
quantifier  $\exists q_i$ is in the same quantifier block as
$d_i$ and $e_i$ and 
thus $\level{q_i}<\level{x_i}$. Consequently, $\sigma_v$
can replace $x_i$ by $q_i$.

\begin{proposition}\label{prop:short-proofs-HKB-in-IRcalcPMDeltaSubst}
Let $\Delta \colon \delta_1,\ldots , \delta_t$ where $\delta_i$ is
$q_i \lor d_i \lor \overline{e}_i, \overline{q}_i \lor \overline{d}_i,
\overline{q}_i \lor e_i$, i.e., $\delta_i$ is the clausal
representation of $q_i \lequiv \overline{d}_i\land e_i$. Let
$\sigma_{v,t}= \{x_i\is q_i \mid 1\leq i\leq t\}$.
There is a \emph{tree} refutation of \hkbt{t} in
\IRcalcPMDSubst{P_t}{M_t}{\Delta}{\sigma_{v,t}}{} of size 
polynomial in $t$.
\end{proposition}

\begin{proof}[sketch]
Derive $\overline{d}_1 \lor \overline{e}_1, \ldots , \overline{d}_t
\lor \overline{e}_t$.  The first clause is derived by a resolution
step between $C_0$ and $C_1$. Then we derive $\overline{d}_{j+1} \lor
\overline{e}_{j+1}$ from $\overline{d}_j \lor \overline{e}_j$,
$C_{2j}\sigma_{v,t}$, $C_{2j+1}\sigma_{v,t}$, and the clauses obtained
from $q_j\lequiv \overline{d}_j \lor e_j$ as follows.
Resolve $d_j \lor \overline{q}_j \lor \overline{d}_{j+1} \lor
\overline{e}_{j+1}$ with $q_j \lor d_j\lor \overline{e}_j$ and derive 
$d_j \lor \overline{e}_j \lor \overline{d}_{j+1} \lor
\overline{e}_{j+1}$ by resolution and factoring. Then continue with 
$\overline{d}_j \lor \overline{e}_j$ and obtain 
$R\colon \overline{e}_j \lor \overline{d}_{j+1} \lor \overline{e}_{j+1}$
by resolution and factoring.
Use $e_j \lor q_j \lor \overline{d}_{j+1} \lor \overline{e}_{j+1}$,
resolve it with $\overline{q}_j\lor e_j$ and factor the resolvent
resulting in $e_j \lor \overline{d}_{j+1} \lor \overline{e}_{j+1}$.
Resolve $R$ with the latter clause, factor the resolvent and obtain
$\overline{d}_{j+1} \lor \overline{e}_{j+1}$.

Each of the $15$ clauses has at most $5$ literals. For $j+1=t$, we have
a similar deduction but with at most $2t+3$ literals per clause. We
obtain $\overline{f}_1\lor \cdots \lor \overline{f}_t$ which can be
resolved by $f_i$ obtained from $\overline{q}_i \lor f_i$ and $q_i
\lor f_i$. Finally, it is easy to check that the refutation has tree
structure and is of size polynomial in $t$.
\qed
\end{proof}
The Herbrand functions obtained from \qres{} or \qures{} refutations
by the method in \cite{DBLP:journals/fmsd/BalabanovJ12} are often
(too) complex. It is easy to check that atomic Herbrand
functions $e_i$ for $x_i$ are sufficient and therefore a short tree 
\IRcalcPMSubst{\cdot}{\cdot}{\cdot} refutation of \hkbt{t} is
possible.  The proof of the following proposition can be found in the
appendix.

\begin{restatable}{prop}{treeRefAtomicFct}
\label{prop:short-proofs-HKB-in-IRcalcPM}
Let $\sigma_{v,t}= \{x_i\is e_i \mid 1\leq i\leq t\}$. Then there is a
\emph{tree} refutation of \hkbt{t} in
\IRcalcPMSubst{P_t}{M_t}{\sigma_{v,t}}{} of size polynomial in $t$.
\end{restatable}

\begin{proposition}\label{prop:IRcalcPM-npsim-IRcalcPMSubst}
\IRcalcPM{\cdot}{\cdot}{} cannot polynomially simulate
\IRcalcPMSubst{\cdot}{\cdot}{\cdot}.
\end{proposition}
\noindent
According to Proposition~\ref{prop:short-proofs-HKB-in-IRcalcPM},
there are not only short tree refutations of \hkbt{t}, 
but also the search space is
limited if a simple heuristic restricting the number of possible
variable replacements $\sigma_{v,t}$ is employed during proof search.
The heuristic requires that for each $(x\is e)\in \sigma_{v,t}$, there
is at least one clause $C\sigma_{v,t}$, which contain duplicate
literals.

\section{Conclusion}\label{sec:concl}

We studied various calculi for QBFs with respect to their relative
strength. We provided polynomial simulations using first-order
translations in order to clarify the possibility to employ
(non-trivial) instantiations in refutations.  By a simulation of
\qres{} and \qures{} by \Rfo, we have seen that the former ones avoid
instantiations.  The simulation of simple instantiation-based calculi
by \Rfo{} revealed that instantiation of universal variables is
possible by resolutions with $p(f_1)$ and $\neg p(f_0)$ together with
the usual propagation of substitutions, and clarified the purpose of
the employed framework of assignments and annotated clauses. We showed
that enabling instantiations with existential variables and formulas
increase the strength of instantiation-based calculi. For
presentational reasons, we have chosen a rather simple approach where
$\sigma_v$ and $\Delta$ are initially given, but it is possible in the
underlying framework to generate $\sigma_v$ and $\Delta$ dynamically.

\medskip
\noindent
\emph{Open problems and future research directions:} 
In all our comparisons, we did not optimize the quantifier prefix by
(advanced) dependency schemes. It is well known that less dependencies
between variables can considerably shorten proofs, for which reason
one would like to integrate these techniques into calculi. We have
left open some proof-theoretical comparisons like sequent systems for
prenex formulas with propositional cuts and
\IRcalcPMDSubst{\cdot}{\cdot}{\cdot}{\cdot} or \IRMcalc{}
\cite{DBLP:conf/mfcs/BeyersdorffCJ14} with our new calculi or
\Rfo. The problem here is that \Rfo{} is probably not strong enough
because inference rules for Skolem function manipulation
\cite{Eder:1992,DBLP:books/el/RV01/BaazEL01} are not available but
seem to be necessary for a polynomial simulation.  The ultimate goal
is to make instantiation-based calculi ready for proof search.  A
first step has been accomplished by showing (in the simulation) that
unrestricted instantiations in \IRcalcPM{\cdot}{\cdot} can be
restricted to minimal ones by simply using unification and mgus like
in the first-order case.  Achieving the goal for strong cacluli is not
an easy exercise because some techniques like extensions are hard to
control.


\bibliographystyle{plain}

\begin{thebibliography}{10}

\bibitem{DBLP:books/el/RV01/BaazEL01}
M.~Baaz, U.~Egly, and A.~Leitsch.
\newblock Normal form transformations.
\newblock In J.~A. Robinson and A.~Voronkov, editors, {\em Handbook of
  Automated Reasoning}, pages 273--333. Elsevier and MIT Press, 2001.

\bibitem{DBLP:journals/fmsd/BalabanovJ12}
V.~Balabanov and J.-H.~R. Jiang.
\newblock Unified {QBF} certification and its applications.
\newblock {\em Formal Methods in System Design}, 41(1):45--65, 2012.

\bibitem{BWJ:SAT14}
V.~Balabanov, M.~Widl, and J.-H.~R. Jiang.
\newblock {QBF} resolution systems and their proof complexities.
\newblock In {\em SAT}, 2014.

\bibitem{DBLP:conf/mfcs/BeyersdorffCJ14}
O.~Beyersdorff, L.~Chew, and M.~Janota.
\newblock On unification of {QBF} resolution-based calculi.
\newblock In E.~Csuhaj{-}Varj{\'{u}}, M.~Dietzfelbinger, and Z.~{\'{E}}sik,
  editors, {\em Mathematical Foundations of Computer Science 2014 - 39th
  International Symposium, {MFCS} 2014, Budapest, Hungary, August 25-29, 2014.
  Proceedings, Part {II}}, volume 8635 of {\em Lecture Notes in Computer
  Science}, pages 81--93. Springer, 2014.

\bibitem{DBLP:conf/stacs/BeyersdorffCJ15}
O.~Beyersdorff, L.~Chew, and M.~Janota.
\newblock Proof complexity of resolution-based {QBF} calculi.
\newblock In E.~W. Mayr and N.~Ollinger, editors, {\em 32nd International
  Symposium on Theoretical Aspects of Computer Science, {STACS} 2015, March
  4-7, 2015, Garching, Germany}, volume~30 of {\em LIPIcs}, pages 76--89.
  Schloss Dagstuhl - Leibniz-Zentrum fuer Informatik, 2015.

\bibitem{BCJ-AAAI-WS16}
O.~Beyersdorff, L.~Chew, and M.~Janota.
\newblock Extension variables in {QBF} resolution.
\newblock In {\em AAAI-16 workshop Beyond NP}, 2016.

\bibitem{DBLP:journals/aml/CookM05}
S.~A. Cook and T.~Morioka.
\newblock Quantified propositional calculus and a second-order theory for
  {NC}$^{\mbox{1}}$.
\newblock {\em Arch. Math. Log.}, 44(6):711--749, 2005.

\bibitem{Eder:1992}
E.~Eder.
\newblock {\em Relative complexities of first order calculi}.
\newblock Artificial intelligence = K{\"u}nstliche Intelligenz. Vieweg, 1992.

\bibitem{DBLP:conf/sat/Egly12}
U.~Egly.
\newblock On sequent systems and resolution for {QBFs}.
\newblock In A.~Cimatti and R.~Sebastiani, editors, {\em SAT}, volume 7317 of
  {\em Lecture Notes in Computer Science}, pages 100--113. Springer, 2012.

\bibitem{ESW-cj-2009}
U.~Egly, M.~Seidl, and S.~Woltran.
\newblock A solver for {QBF}s in negation normal form.
\newblock {\em Constraints}, 14(1):38--79, 2009.

\bibitem{DBLP:conf/cp/Gelder12}
A.~Van Gelder.
\newblock Contributions to the theory of practical quantified boolean formula
  solving.
\newblock In M.~Milano, editor, {\em CP}, volume 7514 of {\em Lecture Notes in
  Computer Science}, pages 647--663. Springer, 2012.

\bibitem{DBLP:journals/tcs/Haken85}
A.~Haken.
\newblock The intractability of resolution.
\newblock {\em Theor. Comput. Sci.}, 39:297--308, 1985.

\bibitem{DBLP:journals/iandc/BuningKF95}
H.~Kleine~B{\"u}ning, M.~Karpinski, and A.~Fl{\"o}gel.
\newblock Resolution for quantified {B}oolean formulas.
\newblock {\em Inf. Comput.}, 117(1):12--18, 1995.

\bibitem{Krajicek:1995}
J.~Kraj\'{i}\v{c}ek.
\newblock {\em Bounded Arithmetic, Propositional Logic, and Complexity Theory},
  volume~60 of {\em Encyclopedia of Mathematics and its Application}.
\newblock Cambridge University Press, 1995.

\bibitem{Leitsch-resolution1997}
A.~Leitsch.
\newblock {\em The resolution calculus}.
\newblock Texts in theoretical computer science. Springer, 1997.

\bibitem{DBLP:journals/jsc/PlaistedG86}
D.~A. Plaisted and S.~Greenbaum.
\newblock A structure-preserving clause form translation.
\newblock {\em J. Symb. Comput.}, 2(3):293--304, 1986.

\bibitem{DBLP:conf/cade/SeidlLB12}
M.~Seidl, F.~Lonsing, and A.~Biere.
\newblock qbf2epr: A tool for generating {EPR} formulas from {QBF}.
\newblock In P.~Fontaine, R.~A. Schmidt, and S.~Schulz, editors, {\em
  PAAR@IJCAR}, volume~21 of {\em EPiC Series}, pages 139--148. EasyChair, 2012.

\bibitem{SlivovskySzeider-SAT14}
F.~Slivovsky and S.~Szeider.
\newblock Variable dependencies and {Q}-resolution.
\newblock In C.~Sinz and U.~Egly, editors, {\em Theory and Applications of
  Satisfiability Testing - {SAT} 2014 - 17th International Conference, Held as
  Part of the Vienna Summer of Logic, {VSL} 2014, Vienna, Austria, July 14-17,
  2014. Proceedings}, volume 8561 of {\em Lecture Notes in Computer Science},
  pages 269--284. Springer, 2014.

\bibitem{Tseitin:1968}
G.~S. Tseitin.
\newblock {O}n the {C}omplexity of {D}erivation in {P}ropositional {C}alculus.
\newblock In A.~O. Slisenko, editor, {\em Studies in Constructive Mathematics
  and Mathematical Logic, Part II}, pages 234--259. Seminars in Mathematics,
  V.A.\ Steklov Mathematical Institute, vol.\ 8, Leningrad, 1968.

\bibitem{DBLP:conf/iccad/ZhangM02}
L.~Zhang and S.~Malik.
\newblock Conflict driven learning in a quantified boolean satisfiability
  solver.
\newblock In L.~T. Pileggi and A.~Kuehlmann, editors, {\em Proceedings of the
  2002 {IEEE/ACM} International Conference on Computer-aided Design, {ICCAD}
  2002, San Jose, California, USA, November 10-14, 2002}, pages 442--449. {ACM}
  / {IEEE} Computer Society, 2002.

\end{thebibliography}

\newpage
\appendix

\section{Proof of some propositions and theorems}

\QBFIsomorphic*

\begin{proof}
The proof is by induction on the logical complexity,
$\lc{\Phi}$, of $\Phi$.

\medskip
\noindent
\base: $\lc{\Phi} = 0$. Then $\Phi$ is $\falsum$, $\verum$ or a Boolean
variable $q$ and $\QBFtoPL{\Phi}{p}{f}$ is $p(f_{0})$,  
$p(f_{1})$ or $p(x_q)$. Then $\Phi \cong \QBFtoPL{\Phi}{p}{f}$.

\medskip
\noindent
\ih: For all QBFs $\Psi$ with $\lc{\Psi} < k$, $\Psi\cong
\QBFtoPL{\Psi}{p}{f}$, i.e., $\Psi$ and $\QBFtoPL{\Psi}{p}{f}$ are isomorphic.

\medskip
\noindent
\step: Consider QBF $\Phi$ with $\lc{\Phi}=k$. In all cases below, 
$\Phi_i\cong \QBFtoPL{\Phi_i}{p}{f}$ holds
($i=1,2$) by the induction hypothesis.

\medskip
\noindent
\sccase{1}: $\Phi = \neg \Phi_1$. Since $\Phi_1 \cong \QBFtoPL{\Phi_1}{p}{f}$,
$\neg \Phi_1 \cong \neg \QBFtoPL{\Phi_1}{p}{f} = \QBFtoPL{\neg \Phi_1}{p}{f}$
and therefore $\Phi \cong \QBFtoPL{\Phi}{p}{f}$ holds.

\medskip
\noindent
\sccase{2}: $\Phi = \Phi_1 \circ \Phi_2$. Since $\Phi_1\cong
\QBFtoPL{\Phi_1}{p}{f}$ as well as $\Phi_2 \cong \QBFtoPL{\Phi_2}{p}{f}$,
$\Phi_1\circ \Phi_2 \cong \QBFtoPL{\Phi_1}{p}{f} \circ \QBFtoPL{\Phi_2}{p}{f}
= \QBFtoPL{\Phi_1\circ \Phi_2}{p}{f}$ and therefore $\Phi \cong
\QBFtoPL{\Phi}{p}{f}$ holds.

\medskip
\noindent
\sccase{3}: $\Phi = \quantifier^b q \, \Phi_1$. Since $\Phi_1 \cong
\QBFtoPL{\Phi_1}{p}{f}$, $\quantifier^b q\, \Phi_1 \cong \quantifier q\,
\QBFtoPL{\Phi_1}{p}{f} = \QBFtoPL{\quantifier^b q\, \Phi_1}{p}{f}$ and
therefore $\Phi \cong \QBFtoPL{\Phi}{p}{f}$ holds.
\qed
\end{proof}

\SatEquiv*

\begin{proof}[sketch]
\noindent
$\Longrightarrow$: $\Phi$ is satisfiable.  We show that
$\QBFtoPL{\Phi}{p}{f}\land p(f_{1}) \land \neg p(f_{0})$ has a model
with a two-element domain $\caU=\{f_{1}, f_{0}\}$ and constants are
mapped to itself by the interpretation function. Moreover, $p(f_{1})$
has to be true and $p(f_{0})$ has to be false.  If we evaluate $\Phi$
according to the semantics, we can, in a parallel way, expand
$\QBFtoPL{\Phi}{p}{f}$ over $\caU$ and obtain two isomorphic expanded
formulas. Evaluating isomorphic leaves in the same way and propagating
the truth values from the leaves to the root (in the corresponding
formula trees) yields the same evaluation result for both formulas.
Hence, $\QBFtoPL{\Phi}{p}{f}\land p(f_{1}) \land \neg p(f_{0})$ is
satisfiable.

\noindent
$\Longleftarrow$: $\Phi$ is unsatisfiable. Then there is a logically
equivalent PCNF $\Phi'$ and a \qres{} refutation of $\Phi'$
(because \qres{} is complete). Due to
Proposition~\ref{prop:isomorphic} and the preservation of the quantifiers and connectives by $\QBFtoPL{\cdot}{p}{f}$, there is an isomorphic PCNF
$\Phi_1'$ of $\QBFtoPL{\Phi}{p}{f}$ where $\Phi_1'$ is logically
equivalent to $\QBFtoPL{\Phi}{p}{f}$. Skolemization yields the
sat-equivalent first-order clause form $\Phi_1''$ of $\Phi_1'$. In
Corollary~\ref{cor:Rfo-psim-qres}, we show that we can simulate
each \qres{} refutation of $\Phi'$ by a first-order resolution
refutation of $\Phi_1'' \land p(f_{1}) \land \neg p(f_{0})$. By
soundness of first-order resolution, we conclude that $\Phi_1''\land
p(f_{1}) \land \neg p(f_{0})$ and therefore $\QBFtoPL{\Phi}{p}{f}\land
p(f_{1}) \land \neg p(f_{0})$ is unsatisfiable.  \qed
\end{proof}

\RfoPsimQU*

\begin{proof}
Let $\Phi\colon \quantifier^b\, M$ be a QBF in PCNF with quantifier
prefix $\quantifier^b$ and matrix $M$.  Consider the first-order
translation $\QBFtoPL{\Phi}{p}{f}$ of $\Phi$ and
$\SkQBFtoPL{\Phi}{p}{f}$ (the skolemized form of
$\QBFtoPL{\Phi}{p}{f}$). By Proposition~2, every literal in $M$ has an
isomorphic counterpart in $\SkQBFtoPL{M}{p}{f}$. We employ this
isomorphism in the following.

Let $C_1, \ldots , C_n$ be a \qures{} deduction of $C_n$.  For any
clause $C_i$ ($1\leq i\leq n$) of the form $L_{i,1}\lor \cdots \lor
L_{i,m_i}$ generate a first-order clause $D_i$ of the form $K_{i,1}\lor
\cdots \lor K_{i,m_i}$ where $K_{i,j} \cong L_{i,j}$ for $j=1,\ldots ,
m_i$.  We show by induction on $n$ that there exists an \Rfo{}
deduction $p(f_{1}), \neg p(f_{0}), E_1,\ldots , E_n$ of $E_n$ from $
\SkQBFtoPL{M}{p}{f} \land p(f_{1}) \land \neg p(f_{0})$ such that the
following holds for all $i=1,\ldots, n$.
\begin{enumerate}
\item \label{item:cond-En-nontaut} $E_i$ is non-tautological.


\item \label{item:cond2-En-more-general} 
$D_i = E_i\sigma$ for some variable substitution $\sigma$.
\end{enumerate}
Condition~\ref{item:cond2-En-more-general} implies that all $E_i$ are
not instantiated with non-variable terms.

\medskip
\noindent
\base: $n=1$. Then $C_1$ is an input clause from $M$, $C_1$ in
non-tautological by assumption (of \qures), and $D_1$ is a first-order
input clause with $C_1\cong D_1$. Take $E_1=D_1$ and $D_1=E_1\sigma$
where $\sigma=\epsilon$.

\medskip
\noindent
\ih: Suppose $n\geq 1$ and for all $k\leq n$, we have based on
$C_1,\ldots , C_k$ and $D_1, \ldots , D_k$ an \Rfo{} deduction
$p(f_{1}), \neg p(f_{0}), E_1,\ldots , E_k$ of $E_k$ from
$\SkQBFtoPL{\Phi}{p}{f} \land p(f_{1}) \land \neg p(f_{0})$ such that
conditions~\ref{item:cond-En-nontaut}.\ and
\ref{item:cond2-En-more-general}.\ hold.

\medskip
\noindent
\step: Consider $C_1, \ldots , C_{n+1}$ and  $D_1, \ldots , D_{n+1}$.

\medskip
\sccase{1}: $C_{n+1}$ is an input clause. Then proceed as in the base case.

\medskip
\sccase{2}: $C_{n+1}$ is the consequence of a $\forall$ reduction
applied to $C_i$ ($i\leq n$). Let $\ell$ be the universal literal
removed. Without loss of generality, let $\ell$ be positive and of the
form $x$.  Then there is a clause $D_i\colon \widetilde{D}_i \lor
p(x)$. Observe that the variable $x$ does not occur in
$\widetilde{D}_i$, because we assume by
Remark~\ref{rem:clauses-as-sets} applications of \Fac{} as early as
possible.
By \ih, we have a non-tautological clause $E_i\colon
\widetilde{E_i}\lor p(y)$ and a variable substitution $\sigma$ with
$D_i = E_i\sigma$.  $E_{n+1}$ is obtained from $E_i$ and $\neg
p(f_{0})$ by resolution resulting in $\widetilde{E_i}$. Then $D_{n+1}
= E_{n+1}\sigma$ and $E_{n+1}$ is non-tautological because $E_i$ is
non-tautological.

\medskip
\sccase{3}: $C_{n+1}$ is a factor of $C_i$ ($i\leq n$).  Then there is
a clause $D_i\colon \widetilde{D}_i \lor \ell(t) \lor \ell(t)$
where $\ell(t)$ is a literal with predicate symbol $p$ with a term $t$
as argument.  By \ih, we have a non-tautological clause $E_i$ and a
variable substitution $\sigma$ with $D_i = E_i\sigma$. If $t$ is a
constant, then $E_{n+1}$ is $E_i$ with one occurrence of $\ell(t)$
removed, $E_i$ is non-tautological and so is $E_{n+1}$ and $D_{n+1} =
E_{k+1}\sigma$.

Let the term $t$ be of the form $f(\vec{X})$. Then $E_i$ is 
$\widetilde{E}_i \lor \ell(f(\vec{Y})) \lor \ell(f(\vec{Z}))$ and
$\sigma(u_r)= x_r$ for all $u_r \in \vec{Y} \cup \vec{Z}$.
Let $\pi$ be the unifier of $\{\ell(f(\vec{Y})), \ell(f(\vec{Z}))\}$
of the form $\{y_i\backslash z_i \mid \text{for all $y_i\in
  \vec{Y}$}\}$. The factor $E_{n+1}$ is then $(\widetilde{E}_i \lor
\ell(f(\vec{Z})))\pi$ and $D_{n+1} = E_{k+1}\sigma$ holds.

We argue in the following that $E_{n+1}$ is non-tautological.  Suppose
$E_{n+1}$ is tautological. Then, since $D_{n+1} = E_{n+1}\sigma$,
$D_{n+1}$ is tautological which in turn implies that $C_{n+1}$ is
tautological. But this is impossible by the definition of \qres{} and
\qures.


Let  $t$ be a variable $x$. Then this case is similar to the case
$t=f(\vec{X})$. 
 
\medskip
\sccase{4}: $C_{n+1}$ is a Q-resolvent of $C_i$ and $C_j$ ($i,j\leq
n$) upon the existential variable $e$.  Then there are two clause
$D_i\colon \widetilde{D}_i \lor p(t_e)$ and $D_j\colon \widetilde{D}_j
\lor \neg p(t_e)$. By \ih, we have non-tautological clauses $E_i$ with
$D_i = E_i\sigma_1$ and $E_j$ with $D_j = E_j\sigma_2$ where
$\sigma_1$ as well as $\sigma_2$ are variable substitutions.

\scsubcase{4}{1}: $t_e$ is a functional term $f_e(\vec{X})$.  Then

\begin{tabular}{rllcl}
$E_i$ & $\colon$ & $\widetilde{E}_i \lor p(f_e(\vec{Y}))$ & &
and 
$\sigma_1(y_r)=x_r$ for all  $y_r \in \vec{Y}$;
\\
$E_j$ & $\colon$ & $\widetilde{E}_j \lor \neg p(f_e(\vec{Z}))$ & &
and 
$\sigma_2(z_r)=x_r$ for all  $z_r \in \vec{Z}$.
\\ 
\end{tabular}

\noindent
Let $\mu$ be a renaming substitution such that $E_i\mu$ and $E_j$ are
variable-disjoint.  In order to construct the resolvent, we need the
mgu $\pi$ of $\{p(f_e(\vec{Y}))\mu, p(f_e(\vec{Z}))\}$, which is
$\{\mu(y_r) \backslash z_r \mid \text{for all $y_r\in
  \vec{Y}$}\}$. The unifier $\pi$ is a matcher; it affects only
variables from $E_i\mu$. The resolvent $E_{n+1}$ is then
$\widetilde{E}_i \mu \pi \lor \widetilde{E}_j$.

We show that there exists a variable substitution $\sigma$ such that
$D_{n+1} = E_{n+1}\sigma$. First observe that $D_i = E_i\mu\sigma_1'$
with $ \sigma_1' = \{\mu(u) \backslash \sigma_1(u) \mid \text{for all
  $u\in \var{E_i}$}\}\setminus \{u\is u\mid \text{$u$ is a
  variable}\}$.  Then with $\sigma_1'' = \{\pi(\mu(u)) \backslash
\sigma_1(u) \mid \text{for all $u\in \var{E_i}$}\}\setminus \{u\is
u\mid \text{$u$ is a variable}\}$, we have $\widetilde{D}_i =
\widetilde{E}_i\mu\pi\sigma_1''$.  For all $y_i\in \vec{Y}$, we have
$\pi(\mu(y_i)) = z_i$, $\sigma_1(y_i) = x_i$ and $\sigma_2(z_i) =
x_i$.
Then
\begin{eqnarray*}
\widetilde{D}_i \lor \widetilde{D}_j  
& = & 
\widetilde{E}_i\mu\pi\sigma_1''
\lor \widetilde{E}_j\sigma_2
\, \, = \, \, (\widetilde{E}_i\mu\pi \lor \widetilde{E}_j)\sigma
\end{eqnarray*}
where $\sigma$ is obtained from 
$$\{\mu(u) \backslash \sigma_1(u) \mid \text{for all $u\in \var{E_i}
  \setminus \vec{Y}$}\} \cup \{v\backslash \sigma_2(v) \mid \text{for
  all $v\in\var{E_j}\}$}\}
$$ by deleting all elements of the form $u\is u$.  Observe that
$\range{\pi}=\{\vec{Z}\} \subseteq \var{E_j}$ and $\range{\pi}
\subseteq \dom{\sigma_2}$.
Therefore $D_{n+1} = E_{n+1}\sigma$.

\scsubcase{4}{2}: $t_e$ is a constant.  Similar to \scsubcase{4}{1}
but with an empty mgu $\pi$.

The clause $E_{n+1}$ from both subcases is non-tautological by the
same reason as in \sccase{3}.

\medskip
\sccase{5}: $C_{n+1}$ is a Q-resolvent of $C_i$ and $C_j$ ($i,j\leq
k$) upon the universal variable $u$. Similar to \scsubcase{4}{1}.
\qed
\end{proof}

\treeRefAtomicFct*

\begin{proof}[sketch]
Take $\sigma_{v,t}=\{x_i\is e_i \mid 1\leq i\leq t\}$ and derive
$\overline{d}_1 \lor \overline{e}_1, \ldots ,
\overline{d}_t \lor \overline{e}_t$.
The first clause is derived by a resolution step between $C_0$ and
$C_1$. Then we derive $\overline{d}_{j+1} \lor \overline{e}_{j+1}$
from $\overline{d}_j \lor \overline{e}_j$ and $C_{2j}\sigma_{v,t}$ and
$C_{2j+1}\sigma_{v,t}$ as follows.
Resolve $\overline{d}_j \lor \overline{e}_j$ and $d_j \lor
\overline{e}_j \lor \overline{d}_{j+1} \lor \overline{e}_{j+1}$,
obtain $\overline{e}_j \lor \overline{e}_j \lor \overline{d}_{j+1}
\lor \overline{e}_{j+1}$ and factor it to get $R\colon \overline{e}_j
\lor \overline{d}_{j+1} \lor \overline{e}_{j+1}$. Next factor $e_j
\lor e_j \lor \overline{d}_{j+1} \lor \overline{e}_{j+1}$ and get $e_j
\lor \overline{d}_{j+1} \lor \overline{e}_{j+1}$. Resolve the latter
with $R$ and factor the resolvent. We get
$\overline{d}_{j+1} \lor \overline{e}_{j+1} $
Each of the $8$ clauses has at most $4$ literals. For $j+1=t$, we have
a similar deduction but with at most $2t+2$ literals per clause. We
obtain $\overline{f}_1\lor \cdots \lor \overline{f}_t$ which can be
resolved by the $f_i$ obtained from $\overline{e}_i \lor f_i$ and $e_i
\lor f_i$. Finally, it is easy to check that the refutation has tree
structure and is of size polynomial in $t$.
\qed
\end{proof}

\end{document}